\documentclass[a4paper,11pt]{article} 
\pdfoutput=1
\usepackage{fullpage}
\usepackage[utf8]{inputenc}
\usepackage[T1]{fontenc}
\usepackage{graphicx}
\usepackage{amsmath}
\usepackage{float}
\usepackage{amsfonts}
\usepackage[dvipsnames,usenames]{color}
\usepackage[colorlinks=true,urlcolor=Blue,citecolor=Green,linkcolor=BrickRed]{hyperref}
\usepackage[usenames,dvipsnames]{xcolor}
\usepackage[noadjust]{cite}
\usepackage{amsmath}
\usepackage{bm}
\usepackage{authblk}
\usepackage{amsthm}

\newcommand{\successor}{\textsf{succ}}

\newtheorem{theorem}{Theorem}
\newtheorem{lemma}[theorem]{Lemma}

\title{Distinct Squares in Circular Words}
\author[1,2]{Mika Amit}
\author[2,3]{Paweł Gawrychowski}
\affil[1]{IBM Research, Haifa, Israel}
\affil[2]{University of Haifa, Israel}
\affil[3]{University of Wrocław, Poland}

\date{}

\begin{document}

\maketitle

\begin{abstract}
A circular word, or a necklace, is an equivalence class under conjugation of a word.
A fundamental question concerning regularities in standard words is bounding the
number of distinct squares in a word of length $n$. The famous conjecture attributed
to Fraenkel and Simpson is that there are at most $n$ such distinct squares, yet the
best known upper bound is $1.84n$ by Deza et al. [Discr. Appl. Math. 180, 52-69 (2015)].
We consider a natural generalization of this
question to circular words: how many distinct squares can there be in all cyclic
rotations of a word of length $n$? We prove an upper bound of $3.14n$. This
is complemented with an infinite family of words implying a lower bound of $1.25n$.
\end{abstract}

\section{Introduction}

Combinatorics on words is mostly concerned with regularities in words. The most basic
example of such a regularity is a square, that is, a substring of the form $uu$.
We might either want to create words with no such substrings, called square-free, or
show that there cannot be too many distinct squares for an arbitrary word of length
$n$. Fraenkel and Simpson proved that $2n$ is an upper bound on the number
of distinct squares contained in a word of length $n$, and also constructed
an infinite family of words of length $n$ containing $n-\Theta(\sqrt{n})$ distinct
squares~\cite{FraenkelS98}. Their upper bound uses a combinatorial
lemma of Crochemore and Rytter~\cite{CrochemoreR95}, called the Three Squares Lemma.
Later, Ilie provided a short and self-contained argument~\cite{Ilie05}.
The Three Squares Lemma is concerned with the rightmost occurrence of every
distinct square, and says that, for any position in the word, there do not exist
three such rightmost occurrences starting at that position (hence the name of the lemma).
It is widely believed that the example given by Frankel and Simpson is the worst possible,
and the right bound is $n$ instead of $2n$. The best known upper bound
was $2n-\Theta(\log n)$~\cite{Ilie2007note} until recently Deza, Franek and Thierry
improved the upper bound to $11/6n$ through a somewhat involved argument~\cite{DezaFT15}.
All these bounds are based on the idea of looking at three rightmost occurrences of squares starting
at the same position. It is known that two such occurrence already imply a certain
periodic structure~\cite{KopylovaS12,FranekFSS12,BlandS15,FanPST06,Simpson07},
and that it is enough to consider binary words~\cite{ManeaS15}.

Regularities are commonly considered in more general contexts than standard words,
such as partial words~\cite{Blanchet-SadriMS09} or trees~\cite{GawrychowskiKRW15,CrochemoreIKKRRTW12}. Another natural
generalization of standard words, motivated by the circular structure of some
biological data, are circular words (also known as necklaces).
A circular word $(w)$ is defined as an equivalence class under conjugation of a word $w$,
that is, it corresponds to all possible rotations of $w$.
Both algorithmic~\cite{HegedusN14,CrochemoreFMP16,CastiglioneRS09}
and combinatorial aspects of such words have been studied.
The latter are mostly motivated by an old result of Thue~\cite{thue1906},
who showed that there is an infinite square-free word over $\{0,1,2\}$.
This started a long line of research of pattern avoidance.
Currie and Fitzpatrick~\cite{CurrieF02} generalized this to circular
words, and then Currie~\cite{Currie02} showed that for any $n\geq 18$ there exists a circular square-free
word of length $n$ (see also a later proof by Shur~\cite{Shur10}).
Recently, Simpson~\cite{Simpson14} considered bounding the number of distinct palindromes in a circular
word of length $n$. It is well-known (and easy to prove) that the number of distinct palindromes
in a standard word of length $n$ is at most $n$. Interestingly, this increases to $5/3n$
for circular words.
Also equations on circular words have been studied~\cite{MasseBGL11}.

We consider the following question: how many distinct squares can there be in a circular
word of length $n$? Note that due to how we have defined a circular word, we are interested
in squares of length at most $n$.
Recall that the $2n$ bound of Fraenkel and Simpson~\cite{FraenkelS98} is based
on the notion of rightmost occurrences. The improved $11/6n$ bound of Deza et al.~\cite{DezaFT15}
is also based on this concept. For a circular word, it is not clear what the rightmost occurrence
might mean, and indeed the proofs seem to completely break. Of course, to bound
the number of distinct squares in a circular word $w$ of length $n$, one can simply bound the
number of distinct squares in a word $ww$ of length $2n$, thus immediately obtaining an
upper bound of $4n$ (by invoking the simple proof of Ilie~\cite{Ilie05}) or
$3.67n$ (by invoking the more involved proof of Deza et al.~\cite{DezaFT15}). This,
however, completely disregards the cyclic nature of the problem.

We start with exhibiting an infinite family of circular words of length $n$ containing
$1.25n-\Theta(1)$ distinct squares. Therefore, it appears that the structure of distinct
squares in circular words is more complex than in standard words. We then continue
with a simple and self-contained upper bound of $3.75n$ on the number of distinct
squares in a circular word of length $n$. Then, by invoking some of the machinery used
by Deza et al.~\cite{DezaFT15}, we improve this to $3.14n$.

\section{Preliminaries}

Let $|w|$ denote the length of a string $w$, $w[i]$ is the $i$-th character of $w$, and $w[i..j]$ is a shortcut for $w[i]w[i+1]\ldots w[j]$.
A natural number $p$ is a period of $w$ iff $w[i]=w[i+p]$ for every $i=1,2,\ldots,|w|-p$.
The smallest such $p$ is called the period of $w$.
We say that $w$ is periodic if its period is at most $|w|/2$, otherwise $w$ is aperiodic.
The well-known periodicity lemma says that if $p$ and $q$ are both periods of $w$ and furthermore
$p+q\leq |w|+\gcd(p,q)$ then $\gcd(p,q)$ is also a period of $w$~\cite{FineWilf}.

$w^{(i)}$ denotes the cyclic rotation of $w$ by $i$, that is, $w[i..|w|]w[1..(i-1)]$.
A circular word $(w)$ is an equivalence class under conjugation of $w$, that is, all cyclic rotations
$w^{(i)}$. A word $uu$ is called a square, and we say that it occurs in $(w)$ if it occurs in $w^{(i)}$
for some $i$. We are interested in bounding the number of distinct squares occurring in a circular
word of length $n$.

\section{Lower bound}

We define an infinite family of words $f_k = \texttt{a}(\texttt{ba})^{k+1}\texttt{a}(\texttt{ba})^{k+2}\texttt{a}(\texttt{ba})^{k+1}\texttt{a}(\texttt{ba})^{k+2}$. See Table~\ref{tbl:family} for an example.
Observe that $|f_k| = 8k+16$. We claim that cyclic rotations of $f_k$ contain many distinct squares.

\begin{figure}[t]
\centering
\begin{tabular}{l l c}
\hline
$ {\bm k}$ & ${\bm f_{\bm k}}$ &  {\bf \#squares} / ${|\bm f_{\bm k}|}$ \\
\hline
1 & \texttt{ababaabababaababaabababa} & 25/24 \\
2 & \texttt{abababaababababaabababaababababa} & 36/32 \\
3 & \texttt{ababababaabababababaababababaabababababa} & 45/40 \\
4 & \texttt{abababababaababababababaabababababaababababababa} & 56/48 \\
5 & \texttt{ababababababaabababababababaababababababaabababababababa} & 65/56 \\
\hline
\end{tabular}
\caption{The number of distinct squares in $f_{k}$, for $k=1,2,3,4,5$.}
\label{tbl:family}
\end{figure}

\begin{lemma}
\label{lem:lowerbound}
For any $k\geq 0$, the circular word $(f_{k})$ contains $10k+16-(k\bmod 2)$ distinct squares.
\end{lemma}

\begin{proof}
To  count distinct squares $uu$ occurring in $(f_k)$, we consider a few disjoint
cases. We first count $uu$ such that $aa$ occurs at most once inside:
\begin{enumerate}
\item Any $uu$ such that $\texttt{aa}$ does not occur inside must be be fully contained in an occurrence of
$\texttt{a}(\texttt{ba})^{k+2}$ or $\texttt{a}(\texttt{ba})^{k+1}$ in $f_{k}$. Thus, to count
such $uu$ we only have to find all distinct squares in $\texttt{a}(\texttt{ba})^{k+2}$.
For any $i=1,2,\ldots,\lfloor (k+2)/2\rfloor$, $(\texttt{ab})^i(\texttt{ab})^i$ and
$(\texttt{ba})^i(\texttt{ba})^i$ appear there, and it can be seen that there are
no other squares. Thus, the number of such $uu$ is exactly $2\lfloor (k+2)/2\rfloor$.
\item Any $uu$ such that $\texttt{aa}$ occurs exactly once inside must have the property
that $u$ starts and ends with $\texttt{a}$. It follows that such $uu$ must be be fully
contained in an occurrence of $\texttt{a}(\texttt{ba})^{k+1}\texttt{a}(\texttt{ba})^{k+1}$
in $f_{k}$. For any $i=0,1,\ldots,k+1$, $\texttt{a}(\texttt{ba})^{i}\texttt{a}(\texttt{ba})^{i}$
appears there, and it can be seen that there are no other squares containing exactly
one occurrence of $\texttt{aa}$, so there are exactly $k+2$ such $uu$.
\end{enumerate}

Then we count $uu$ such that $\texttt{aa}$ occurs exactly twice inside. Then, $\texttt{aa}$
must occur once in $u$ and furthermore, by analyzing the distances between the occurrences of
$\texttt{aa}$ in $f_{k}$, we obtain that $|u|=2k+5$ or $|u|=2k+3$. We analyze these
two possibilities:
\begin{enumerate}
\item If $|u|=2k+3$ then $uu$ appears in an occurrence of 
$(\texttt{ba})^k\texttt{baa}(\texttt{ba})^k\texttt{baa}(\texttt{ba})^k\texttt{b}$ in $f_{k}$.
There are $2k+2$ such $uu$.
\item If $|u|=2k+5$ then $uu$ appears in an occurrence of
$\texttt{a}(\texttt{ba})^k\texttt{baaba}(\texttt{ba})^{k}\texttt{baaba}(\texttt{ba})^{k}$ in $f_{k}$.
There are $2k+2$ such $uu$.
\end{enumerate}

Finally, we count $uu$ such that $\texttt{aa}$ occurs at least three times inside. By analyzing
the distances between the occurrences of $\texttt{aa}$ in $f_{k}$, we obtain that in such case
$|u|=4k+8$, so $|uu|=|f_{k}|$. We claim that there are
exactly $|f_k|/2=4k+8$ such $uu$. To prove this, write $f_k=x_k x_k$ with $x_k=\texttt{a}(\texttt{ba})^{k+1}\texttt{a}(\texttt{ba})^{k+2}$.
$x_k$ cannot be represented as a nontrivial power $y^p$ with $p\geq 2$, because $\texttt{aa}$
occurs only once inside $x_k$, so it would mean that $y$ starts and ends with $\texttt{a}$,
but then $p=2$ is not possible due to $|\texttt{a}(\texttt{ba})^{k+1}|\neq |\texttt{a}(\texttt{ba})^{k+2}|$,
and $p\geq 3$ would generate another occurrence of $\texttt{a}$.
Clearly, every cyclic shift of $f_k$ is a
square occurring in $(f_k)$, because a cyclic shift of a square is still a square. It remains
to count distinct cyclic shifts of $f_k$. Assume that two of these shifts are equal,
that is, $(f_k)^{(i)}=(f_k)^{(j)}$ for some $0\leq i < j < |f_k|$, so $x_k = (x_k)^{(j-i)}$.
Then $\gcd(|x_k|,j-i)$ is a period of $x_k$. But $x_k$ is not a nontrivial power,
so $j-i =0 \bmod |x_k|$. Consequently, every $i=0,1,\ldots,|x_k|-1$ generates
a distinct square.

All in all, the number of distinct squares occurring in $(f_{k})$ is
\[ k+2+2\lfloor (k+2)/2\rfloor+2(2k+2)+4k+8 = 9k+16+2\lfloor k/2\rfloor \]
or, in other words, $10k+16-(k \bmod 2)$.
\end{proof}

By Lemma~\ref{lem:lowerbound}, for any $n_{0}$ there exists a circular word of length $n\geq n_{0}$
containing at least $1.25n-\Theta(1)$ distinct squares.

\section{Upper bound}

Our goal is to upper bound the number of distinct squares occurring in a circular word $(w)$
of length $n$. Each such square occurs in $ww$, hence clearly
there are at most $4n$ such distinct squares by plugging in the known bound on the
number of distinct squares. However, we want a stronger bound.

Recall that the bound on the number of distinct squares is based on the notion
of the rightmost occurrence. For every distinct square $uu$ occurring in a word, we choose its
rightmost occurrence. Then, we have the following property.

\begin{lemma}[\cite{FraenkelS98}]
For any position $i$, there are at most two rightmost occurrences starting
at $i$.
\end{lemma}

Consider the rightmost occurrences of distinct squares of length up to $n$ in $ww$.
We first analyze the rightmost occurrences starting at positions $1,2,\ldots,\frac{1}{4}n$.

\begin{lemma}
\label{lem:quarter}
If $w[\frac{1}{4}n..\frac{1}{2}n]$ is aperiodic then every rightmost occurrence starting
at position $i\in\{1,2,\ldots,\frac{1}{4}n\}$ is of the same length.
\end{lemma}

\begin{proof}
Assume otherwise, that is, $w[\frac{1}{4}n..\frac{1}{2}n]$ is aperiodic, but there
are two rightmost occurrences $uu$ and $u'u'$ starting at positions $i,i'\in\{1,2,\ldots,\frac{1}{4}n\}$,
respectively, in $ww$ such that $|u| > |u'|$.
Then, $i+2|u| > n$ and $i'+2|u'|>n$, as otherwise we could have found the
same square in the second half of $ww$. Because $|u|,|u'| \leq \frac{1}{2}n$, this implies
$i+|u| > \frac{1}{2}n$ and $i'+|u'| > \frac{1}{2}n$. So $w[\frac{1}{4}n..\frac{1}{2}n]$
\footnote{Formally, we need to appropriately round both $\frac{1}{4}n$ and $\frac{1}{2}n$. We chose not
to do so explicitly as to avoid cluttering the presentation.}
is fully inside the first half of both $uu$ and $u'u'$. But then it also appears starting
at positions $\frac{1}{4}n+|u|$ and $\frac{1}{4}n+|u'|$, see Figure~\ref{fig:move}.
The distance between these two distinct (due to $|u|>|u'|$) occurrences is
\[ (\frac{1}{4}n+|u|) - (\frac{1}{4}n+|u'|) = |u| - |u'| \]
We know that $|u|\leq \frac{1}{2}n$ and $|u'|>\frac{1}{2}n-i' \geq \frac{1}{2}n-\frac{1}{4}n=\frac{3}{8}n$.
Thus, the distance is less than $\frac{1}{2}n-\frac{3}{8}n = \frac{1}{8}n$
and we conclude that the period of $w[\frac{1}{4}n..\frac{1}{2}n]$ is at most $\frac{1}{8}n$,
which is a contradiction.
\end{proof}

\begin{figure}[t]
\begin{center}
\includegraphics[width=0.8\textwidth]{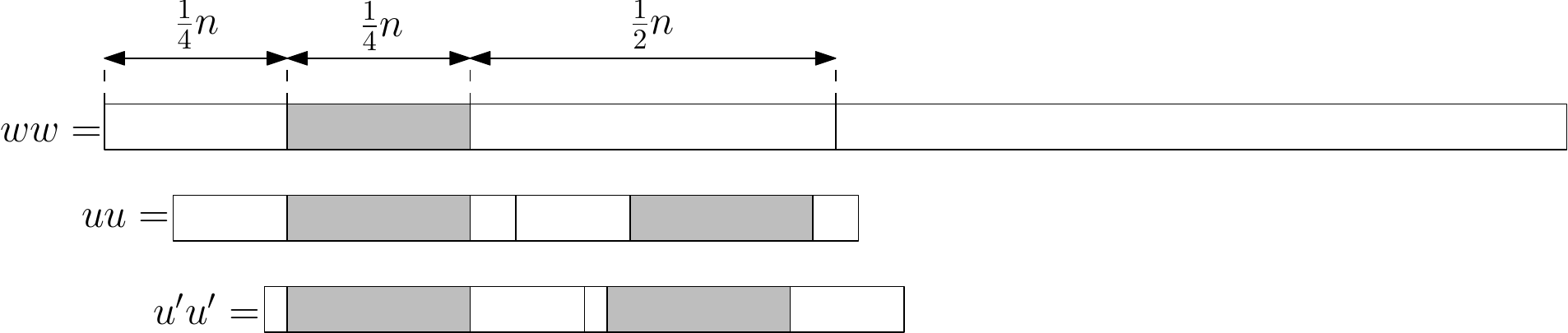}
\end{center}
\caption{Two rightmost occurrences of squares $uu$ and $u'u'$ in $ww$.}
\label{fig:move}
\end{figure}

By Lemma~\ref{lem:quarter}, assuming that $w[\frac{1}{4}n..\frac{1}{2}n]$ is aperiodic,
for every $i=1,2,\ldots,\frac{1}{4}n$ there is at most one rightmost occurrence starting at $i$. For all the remaining $i$,
there are at most two rightmost occurrences starting at $i$, making the total
number of distinct squares at most $\frac{1}{4}n+2(2n-\frac{1}{4}n)=3\frac{3}{4}n$.

It might be the case that $w[\frac{1}{4}n..\frac{1}{2}n]$ is periodic. However, the number
of distinct squares occurring in $(w)$ is the same as the number of distinct 
squares occurring in any $(w^{(i)})$, so we are free to replace $w$ with any
of its cyclic shifts. We claim that if, for any $i=0,1,\ldots,n-1$,
$w^{(i)}[\frac{1}{4}n..\frac{1}{2}n]$ is periodic, then the whole $w$ is a nontrivial power
$y^p$ with $p\geq 8$. To show this, we need an auxiliary lemma that is a special case of Lemma 8.1.2
of~\cite{Lothaire2002}. We provide a proof for completeness.

\begin{lemma}
\label{lem:extendperiod}
For any word $w$ and characters $a,b$, if both $aw$ and $wb$ are periodic then their periods
are in fact equal.
\end{lemma}

\begin{proof}
We assume that the period of $aw$ is $p\leq |aw|/2$ and the period of $wb$ is $q\leq |wb|/2$.
Then $p$ and $q$ are both periods of $w$. By symmetry, we can assume that $p\geq q$.
$p+q\leq (|aw|+|wb|)/2=1+|w|$, so by the periodicity lemma $\gcd(p,q)$ is a period of $w$. 
We claim that $\gcd(p,q)$ is also a period of $aw$. To prove this, it is enough to show that
$a=w[\gcd(p,q)]$. $\gcd(p,q)$ is a period of $w$ and, for $n\geq 2$,
$p\leq |w|$, so this is equivalent to showing that $a=w[p]$. But
this holds due to $p$ being a period of $aw$. Hence $\gcd(p,q)$ is a period of $aw$,
but $p$ is the period of $aw$ and $p\geq q$, therefore $p=q$.
\end{proof}

We observe that the substrings $w^{(i)}[\frac{1}{4}n..\frac{1}{2}n]$ correspond to all substrings
of length $\frac{1}{4}n$ of $ww$. By Lemma~\ref{lem:extendperiod}, if every substring of
length $\frac{1}{4}n$ of $ww$ is periodic, then the periods of all such substrings are the
same and equal to $d\leq \frac{1}{8}n$. Therefore, $d$ is also a period of the whole $ww$.
But then $\gcd(|w|,d)\leq d \leq \frac{1}{8}|w|$ is also a period of $ww$. We conclude that
$\gcd(|w|,d) \leq \frac{1}{8}|w|$ is period of $w$, hence $w=y^p$ for some $p\geq 8$,
as claimed.

It remains to analyze the number of distinct squares in a circular word $(w)$, where $w=y^p$
for $p \geq 8$. Each such square is a distinct square in $y^{p+1}$. The number of distinct
squares in $y^{p+1}$ is at most $2(p+1)|y|  = 2\frac{p+1}{p} n \leq 2.25n$, since $p\geq 8$.

\begin{theorem}
The number of distinct squares in a circular word of length $n$ is at most $3.75n$.
\end{theorem}

To improve on the above upper bound, we need some of the machinery used by Deza et al.~\cite{DezaFT15}.
Two occurrences of squares $uu$ and $UU$ starting at the same position such that $|u|<|U|$
are called a double square and denoted $(u,U)$. If both are the rightmost occurrences,
this is an FS-double square. An FS-double square is identified with the starting position of
the two occurrences.

\begin{lemma}[see proof of Theorem 32 in~\cite{DezaFT15}]
\label{lemma:FS}
If $(u,U)$ is the leftmost FS-double square of a string $x$ and $|x|\geq 10$, then the number
of FS-double squares in $x$ is at most $\frac{5}{6}|x|-\frac{1}{3}|u|$.
\end{lemma}

We again consider the rightmost occurrence of every distinct square of length up to $n$
in $ww$ and assume that $w[\frac{1}{4}n..\frac{1}{2}n]$ is aperiodic (as otherwise we already know
there are at most $2.25n$ distinct squares). We need to consider two cases: either there
are no rightmost occurrences starting at $i=1,2,\ldots,\frac{1}{4}n$, or there is at least one
such occurrence.

\paragraph{No rightmost occurrences starting at $i=1,2,\ldots,\frac{1}{4}n$.} In this case,
it is enough to bound the number of distinct squares in $\hat{w}=w[(\frac{1}{4}n+1)..n]w$.
Let $i$ be the starting position of the leftmost FS-double square $(u,U)$ in $\hat{w}$.
If $i>\frac{3}{4}n$ then the total number of distinct squares is at most $\frac{3}{4}n+2n=2\frac{3}{4}n$,
so we assume $i\leq \frac{3}{4}n$. Then, the total number of distinct squares can be bounded by
applying Lemma~\ref{lemma:FS} on $w[(\frac{1}{4}n+i)..n]w$ to show that the number
of FS-double squares is at most
\[ \frac{5}{6}(\frac{7}{4}n-i+1) - \frac{1}{3} |u| \]
We know that $i+2|u| > \frac{3}{4}n$, as otherwise $uu$ would occur later in $w$. Therefore,
the maximum number of distinct squares is
\begin{equation}
\frac{7}{4}n + \frac{5}{6}(\frac{7}{4}n-i+1) - \frac{1}{3}\frac{\frac{3}{4}n-i+1}{2} = (\frac{7}{4}+\frac{35}{24}-\frac{1}{8})n-(\frac{5}{6}-\frac{1}{6})i+\frac{4}{6} \leq 3\frac{1}{12}n \label{eq:remaining}
\end{equation}

\paragraph{At least one rightmost occurrence starting at $i\in\{1,2,\ldots,\frac{1}{4}n\}$.}
We now move to the more interesting case where there are some rightmost occurrences starting
at $i=1,2,\ldots,\frac{1}{4}n$. We then know by Lemma~\ref{lem:quarter} that they all correspond to squares of the
same length $2\ell$. Let $i\in \{1,2,\ldots,\frac{1}{4}n\}$ be the starting position of one of
these rightmost occurrences. Then, $i+2\ell > n$ as otherwise the square would occur later
in the second $w$, so $\ell > (n-\frac{n}{4})/2 = \frac{3}{8}n$. We also know that $\ell < \frac{1}{2}n$, as otherwise
$w=y^{2}$ and there are only $3n$ distinct squares. To conclude, $\ell \in (\frac{3}{8}n,\frac{1}{2}n)$.
Observe that, due to the square starting at position $i$, the aperiodic substring $s=w[\frac{1}{4}n..\frac{1}{2}n]$
also occurs at position $\frac{1}{4}n+\ell$ in $ww$. Therefore, we can rotate $w$ by $\ell$ 
and repeat the whole reasoning. We either obtain that the number of distinct squares is at most
$3\frac{1}{12}n$ (if, in $w^{(\ell)}w^{(\ell)}$, there are no rightmost occurrences starting at $i=1,2,\ldots,\frac{1}{4}n$), or there is another occurrence of $s$
at position $\frac{1}{4}n+\ell+\ell'-n$ in $w$, where $\ell,\ell' \in (\frac{3}{8}n,\frac{1}{2}n)$.
Because $s$ is aperiodic and $\ell+\ell' > \frac{3}{4}n$, the other occurrence must actually be at
position $\frac{1}{4}n-\Delta$, where $\Delta\in (\frac{1}{8}n,\frac{1}{4}n)$. By repeating
this enough times (and recalling that two occurrences of $s$ cannot be too close to each
other, as otherwise $s$ is not aperiodic), we either obtain that there are at most $3\frac{1}{12}n$
distinct squares or all occurrences of $s$ in $(w)$ are at positions $\frac{1}{4}n+\sum_{j=1}^{i-1}\Delta_{j}$
(recall that $(w)$ denotes the circular word, so we calculate positions modulo $n$) for $i=1,2,\ldots,d$,
where $\sum_{j=1}^{d}\Delta_{j}=n$ and $\Delta_{j}\in (\frac{1}{8}n,\frac{1}{4}n)$ for every $j=1,2,\ldots,d$.
That is, the whole $(w)$ is covered by the occurrences of $s$, and because $s$ is aperiodic these
occurrences overlap by less than $\frac{1}{8}n$. Observe that there cannot be any other occurrences
of $s$ in $(w)$, because the additional occurrence would overlap with one of the already found
occurrences by at least $\frac{1}{8}n$, thus contradiction the assumption that $s$ is aperiodic.
By the constraints on $\Delta_{j}$, $d\in \{5,6,7\}$. See Figure~\ref{fig:cycle} for an illustration
with $d=7$.
We further consider three possible subcases.

\begin{figure}[t]
\begin{center}
\includegraphics[width=0.5\textwidth]{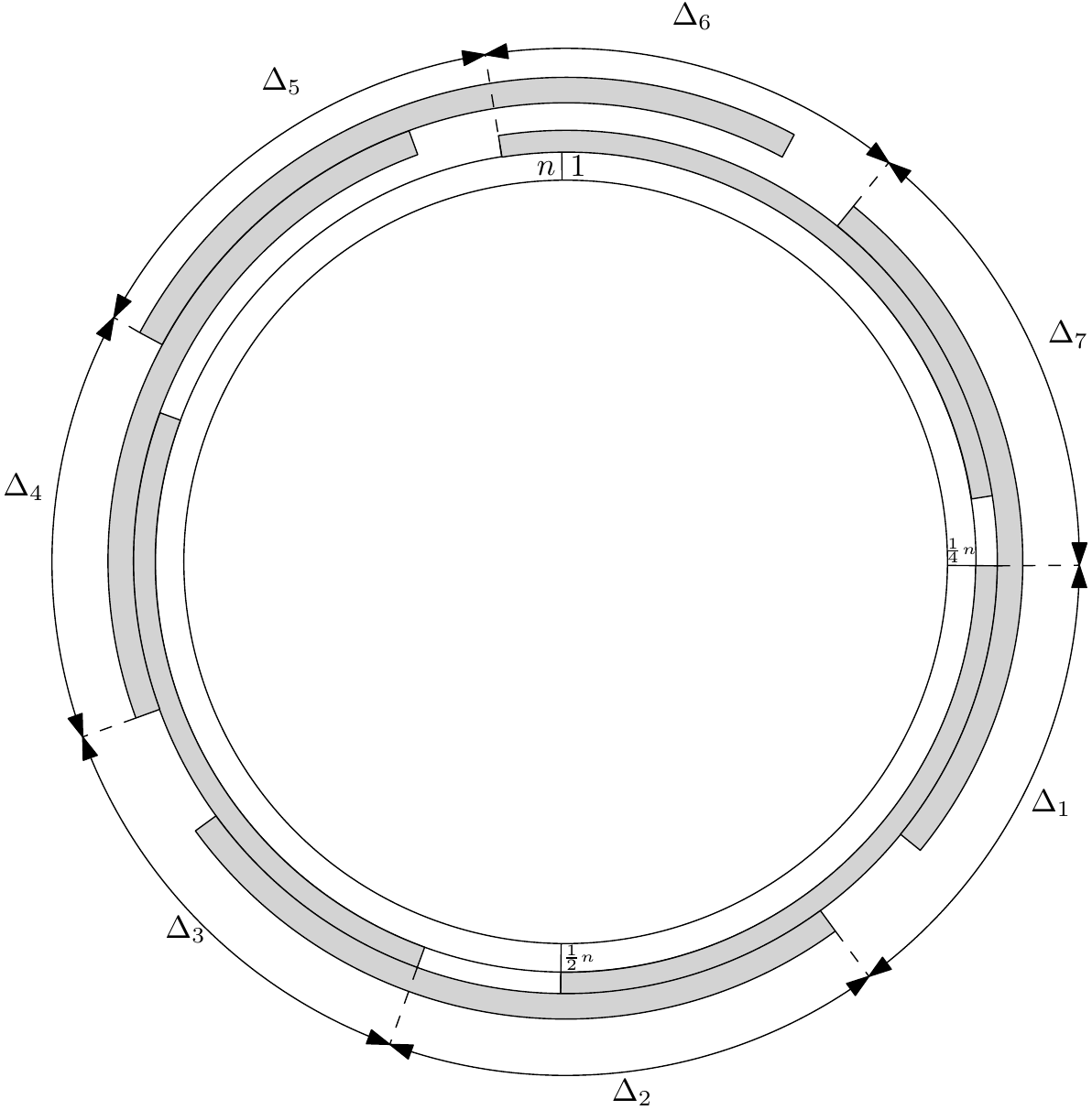}
\end{center}
\caption{Seven occurrences of an aperiodic $s$ of length $\frac{1}{4}n$ inside $(w)$.}
\label{fig:cycle}
\end{figure}

\paragraph{$d=5$. } In such case, we have $\Delta_{j}\geq \frac{1}{5}n$ for some $j$. By rotating
$w$, we can assume that $j=1$. Recall that then all squares starting at $i=1,2,\ldots,\frac{1}{4}n$
have the same length $2\ell$ (and there is at least one such square), so there is another occurrence
of $s$ starting at position $\frac{1}{4}n+\ell$, and then by repeating the reasoning at position
$\frac{1}{4}n+\ell+\ell'$, where $\ell+\ell'=n-\Delta_{1}$ (due to $\ell,\ell'\in (\frac{3}{8}n,\frac{1}{2}n)$).
Combining this with $\Delta_{1}\geq\frac{1}{5}n$, we obtain that $\min\{\ell,\ell'\} \leq \frac{2}{5}n$.
By again rotating $w$, we can assume that in fact $\ell \leq \frac{2}{5}n$.
Let $i\in\{1,2,\ldots,\frac{1}{4}n\}$ be the starting position of a rightmost occurrence of a square 
of length $2\ell$. Then $i+2\ell>n$ as otherwise it would not be a rightmost occurrence,
so $i > \frac{1}{5}n$ and we obtain that there are less than $\frac{1}{4}n-\frac{1}{5}n=\frac{1}{20}n$
rightmost occurrences starting at $i=1,2,\ldots,\frac{1}{4}n$. By the previous calculation~(\ref{eq:remaining})
the number of remaining rightmost occurrences is at most $3\frac{1}{12}n$,
making the total number of distinct squares at most $3\frac{2}{15}n$.

\paragraph{$d=6$.} We will show that this is, in fact, not possible. 
Recall that, for every $i=1,2,\ldots,6$, after rotating $w$ by $r=\sum_{j=1}^{i-1}\Delta_{j}$ 
we obtain that there is at least one rightmost occurrence starting in the prefix of length
$\frac{1}{4}n$ of $w^{(r)}w^{(r)}$, and in fact, by Lemma~\ref{lem:quarter}, all such rightmost occurrences correspond
to squares of the same length $2\ell_{i}$, where $\ell_{i}\in (\frac{3}{8}n,\frac{1}{2}n)$.
Thus, for every occurrence of $s$ starting at position $\frac{1}{4}n+\sum_{j=1}^{i-1}\Delta_{j}$,
there is another occurrence at position $\frac{1}{4}n+\sum_{j=1}^{i-1}\Delta_{j}+\ell_{i}$
in $(w)$ (recall that the positions are taken modulo $n$). We claim that
$\ell_{i}=\Delta_{i}+\Delta_{i+1}$ or $\ell_{i}=\Delta_{i}+\Delta_{i+1}+\Delta_{i+2}$,
where the indices are taken modulo 6. Certainly, $\ell_{i}=\Delta_{i}+\Delta_{i+1}+\ldots+\Delta_{i+k}$
for some $k$. We cannot have $k=0$ because $\ell_{i}>\frac{3}{8}n$ and $\Delta_{i}<\frac{3}{8}n$.
We also cannot have $k\geq 3$, because $\ell_{i}<\frac{1}{2}n$ and
$\Delta_{i}+\Delta_{i+1}+\Delta_{i+2}+\Delta_{i+3} > \frac{1}{2}n$. So,
$k=1$ or $k=2$. For every $i=1,2,\ldots,6$, we define $\successor(i)\in \{1,2,\ldots,6\}$
as follows. If $\ell_{i}=\Delta_{i}+\Delta_{i+1}$ then we set $\successor(i)=i+2$,
and otherwise (if $\ell_{i}=\Delta_{i}+\Delta_{i+1}+\Delta_{i+2}$) $\successor(i)=i+3$.
Intuitively, every occurrence of $s$ in $(w)$ points to another such occurrence.
Due to $\ell_{i}\in (\frac{3}{8}n,\frac{1}{2}n)$ holding for every $i=1,2,\ldots,6$,
the difference between the starting positions of the $i$-th and the $\successor(i)$-th
occurrence of $s$ belongs to $(\frac{3}{8}n,\frac{1}{2}n)$, so the difference between
the starting position of the $i$-th and the $\successor(\successor(i))$-th occurrence
of $s$ belongs to $(\frac{3}{4}n,n)$. In fact, due to $s$ being aperiodic, the latter
difference must belong to $(\frac{3}{4}n,\frac{7}{8}n)$. Consequently, there are no
other occurrences of $s$ between the $\successor(\successor(i))$-th and the $i$-th,
so $\successor(\successor(i))=i-1$.
Now, we consider two cases:
\begin{enumerate}
\item $\successor(1)=3$, then $\successor(3)=6$, so $\successor(6)=2$, $\successor(2)=5$ and $\successor(5)=1$.
\item $\successor(1)=4$, then $\successor(4)=6$, so $\successor(6)=3$, $\successor(3)=5$, $\successor(5)=2$, $\successor(2)=4$.
\end{enumerate}
In both cases, we obtain that $\successor(i)=\successor(j)$ for some $i\neq j$. But this is a contradiction,
because then there are two occurrences of $s$ within distance less than $\frac{1}{8}n$, so $s$
is not aperiodic.

\paragraph{$d=7$.} We define $\successor(i)$ for every $i=1,2,\ldots,7$ as in the previous case.
Because $\successor(i)\in \{i+2,i+3\}$ and $\successor(\successor(i))=i-1$ still holds, we obtain
that in fact $\successor(i)=i+3$ for every $i=1,2,\ldots,7$. This means that $\ell_{i}=\Delta_{i}+\Delta_{i+1}+\Delta_{i+2}$.
Consider all rightmost occurrences starting at $i=1,2,\ldots,\frac{1}{4}n$. We must have
that $i+2\ell_{1}>n$ for each of them, so $i>n-2(\Delta_{1}+\Delta_{2}+\Delta_{3})$, making the total
number of such occurrences at most $\min\{\frac{1}{4}n,2(\Delta_{1}+\Delta_{2}+\Delta_{3})-\frac{3}{4}n\}$.
Because $\Delta_{1}+\Delta_{2}+\Delta_{3}\leq \frac{1}{2}n$ due to $\Delta_{i}>\frac{1}{8}n$ holding for every $i=1,2,\ldots,7$
and $\sum_{i=1}^{7}\Delta_{i}=n$, this number is actually $2(\Delta_{1}+\Delta_{2}+\Delta_{3})-\frac{3}{4}n$.

\begin{figure}[t]
\begin{center}
\includegraphics[width=0.8\textwidth]{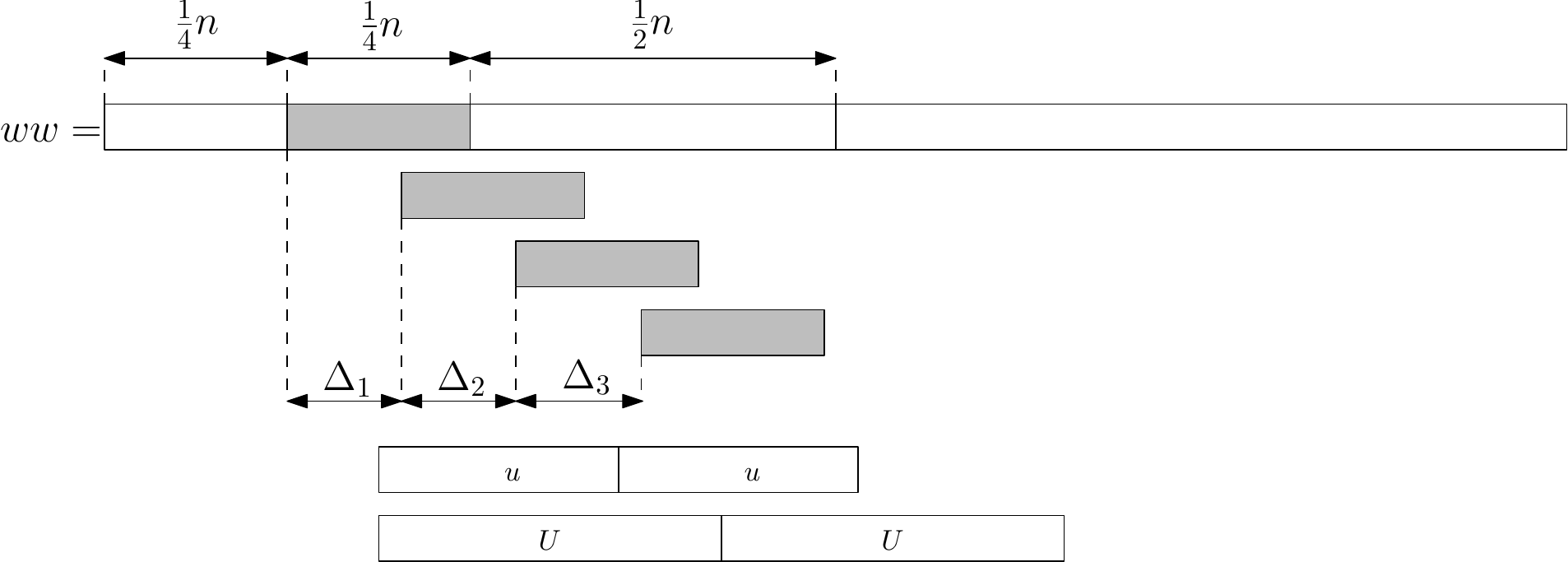}
\end{center}
\caption{The leftmost FS-square starting at position $j\leq\frac{1}{4}n+\Delta_{1}$.}
\label{fig:uu}
\end{figure}

Now we must account for the remaining distinct squares. Let $j$ be the starting position of the leftmost
FS-double square $(u,U)$ in $ww$. Note that $j>\frac{1}{4}n$ because there is at most one rightmost occurrence
starting at $i=1,2,\ldots,\frac{1}{4}n$. We lower bound $j$ by considering two possible cases:
\begin{enumerate}
\item $j > \frac{1}{4}n+\Delta_{1}$.
\item $j \leq \frac{1}{4}n+\Delta_{1}$, then the occurrences of $s$ starting at $\frac{1}{4}n+\Delta_{1}$ and
$\frac{1}{4}n+\Delta_{1}+\Delta_{2}+\Delta_{3}$ are disjoint and both fully inside the first $w$, because
$\Delta_{1}+\Delta_{2}+\Delta_{3}\leq \frac{1}{2}n$. Thus, both $u$ and $U$ contain $s$ as a substring.
See Figure~\ref{fig:uu}.
Then, because all occurrences of $s$ start at positions of the form $\frac{1}{4}n+\sum_{j=1}^{i-1}\Delta_{j}$,
we conclude that $|u|=\Delta_{2}+\Delta_{3}$ and $|U|=\Delta_{2}+\Delta_{3}+\Delta_{4}$.
So, $j > n-2(\Delta_{2}+\Delta_{3})$.
\end{enumerate}
We now know that $j>\min\{\frac{1}{4}n+\Delta_{1},n-2(\Delta_{2}+\Delta_{3})\}$.
Using $j+2|u|>n$ we obtain that the number of remaining distinct squares is at most
\[ 1\frac{3}{4}n+\frac{5}{6}(2n-j)-\frac{1}{3}|u| \leq 3\frac{5}{12}n-\frac{5}{6}j-\frac{1}{3}\frac{n-j}{2} =3\frac{1}{4}n-\frac{2}{3}j\]
so the total number of squares is
\begin{align*}
& \leq 3\frac{1}{4}n+2(\Delta_1+\Delta_2+\Delta_3)-\frac{3}{4}n-\frac{2}{3}j\\
& \leq 2\frac{1}{2}n+ 2(\Delta_{1}+\Delta_{2}+\Delta_{3})-\frac{2}{3}\min\{\frac{1}{4}n+\Delta_{1},n-2(\Delta_{2}+\Delta_{3})\}
\end{align*}
We rewrite the above in terms of $\ell_{1}$ and $\Delta_{1}$:
\[ 2\frac{1}{2}n+ 2\ell_{1}-\frac{2}{3}\min\{\frac{1}{4}n+\Delta_{1},n-2\ell_{1}+2\Delta_{1}\}
\leq 2\frac{1}{2}n+ 2\ell_{1}-\frac{2}{3}\min\{\frac{3}{8}n,\frac{5}{4}n-2\ell_{1}\}
\]
The above expression is increasing in $\ell_{1}$. 
Because $\sum_{i=1}^{7}\ell_{i}=\sum_{i=1}^{7}(\Delta_{i}+\Delta_{i+1}+\Delta_{i+2}) = 3n$, after an appropriate rotation we can assume that $\ell_{1}\leq \frac{3}{7}n$,
and bound the expression:
\[ 2\frac{1}{2}n+\frac{6}{7}n-\frac{2}{3}\min\{\frac{3}{8}n,\frac{5}{4}n-\frac{6}{7}n\}=
3\frac{5}{14}n-\frac{1}{4}n=3\frac{3}{28}n \]

\paragraph{Wrapping up. } We have obtained that either there is an aperiodic substring
of length $\frac{1}{4}n$, and thus there are at most $2.25n$ distinct squares, or
there are no rightmost occurrences starting at $i=1,2,\ldots,\frac{1}{4}n$ and the maximum
number of distinct squares is $3\frac{1}{12}n$, or there is at least at least one rightmost
occurrence starting at $i\in\{1,2,\ldots,\frac{1}{4}n\}$. In the last case,
either $d=5$ and there are at most $3\frac{2}{15}n$ distinct squares, or $d=7$
and there are at most $3\frac{3}{28}n$ distinct squares. The maximum of these
upper bounds is $3\frac{2}{15}n$.

\begin{theorem}
The number of distinct squares in a circular word of length $n$ is at most $3.14n$.
\end{theorem}

\section{Conclusions}

We believe that it should be possible to show an upper bound of $3n$, possibly without using
the machinery of Deza et al., but it seems to require some new combinatorial insights. A
computer search seems to suggest that the right answer is $1.25n$, but showing this
is probably quite difficult. Another natural direction for a follow-up work is to consider higher
powers in circular words.

\bibliographystyle{plain}
\bibliography{biblio}

\begin{thebibliography}{10}

\bibitem{Blanchet-SadriMS09}
Francine Blanchet{-}Sadri, Robert Mercas, and Geoffrey Scott.
\newblock Counting distinct squares in partial words.
\newblock {\em Acta Cybern.}, 19(2):465--477, 2009.

\bibitem{BlandS15}
Widmer Bland and William~F. Smyth.
\newblock Three overlapping squares: The general case characterized {\&}
  applications.
\newblock {\em Theor. Comput. Sci.}, 596:23--40, 2015.

\bibitem{CastiglioneRS09}
Giusi Castiglione, Antonio Restivo, and Marinella Sciortino.
\newblock Circular {S}turmian words and {H}opcroft's algorithm.
\newblock {\em Theor. Comput. Sci.}, 410(43):4372--4381, 2009.

\bibitem{CrochemoreFMP16}
Maxime Crochemore, Gabriele Fici, Robert Mercas, and Solon~P. Pissis.
\newblock Linear-time sequence comparison using minimal absent words {\&}
  applications.
\newblock In {\em {LATIN}}, volume 9644 of {\em Lecture Notes in Computer
  Science}, pages 334--346. Springer, 2016.

\bibitem{CrochemoreIKKRRTW12}
Maxime Crochemore, Costas~S. Iliopoulos, Tomasz Kociumaka, Marcin Kubica, Jakub
  Radoszewski, Wojciech Rytter, Wojciech Tyczynski, and Tomasz Walen.
\newblock The maximum number of squares in a tree.
\newblock In {\em {CPM}}, volume 7354 of {\em Lecture Notes in Computer
  Science}, pages 27--40. Springer, 2012.

\bibitem{CrochemoreR95}
Maxime Crochemore and Wojciech Rytter.
\newblock Squares, cubes, and time-space efficient string searching.
\newblock {\em Algorithmica}, 13(5):405--425, 1995.

\bibitem{Currie02}
James~D. Currie.
\newblock There are ternary circular square-free words of length $n$ for $n\geq
  18$.
\newblock {\em Electr. J. Comb.}, 9(1), 2002.

\bibitem{CurrieF02}
James~D. Currie and D.~Sean Fitzpatrick.
\newblock Circular words avoiding patterns.
\newblock In {\em Developments in Language Theory}, volume 2450 of {\em Lecture
  Notes in Computer Science}, pages 319--325. Springer, 2002.

\bibitem{DezaFT15}
Antoine Deza, Frantisek Franek, and Adrien Thierry.
\newblock How many double squares can a string contain?
\newblock {\em Discrete Applied Mathematics}, 180:52--69, 2015.

\bibitem{FanPST06}
Kangmin Fan, Simon~J. Puglisi, William~F. Smyth, and Andrew Turpin.
\newblock A new periodicity lemma.
\newblock {\em {SIAM} J. Discrete Math.}, 20(3):656--668, 2006.

\bibitem{FineWilf}
N.J. Fine and H.S. Wilf.
\newblock Uniqueness theorems for periodic functions.
\newblock In {\em Proc. Am. Math. Soc.}, volume~16, pages 109--114, 1965.

\bibitem{FraenkelS98}
Aviezri~S. Fraenkel and Jamie Simpson.
\newblock How many squares can a string contain?
\newblock {\em J. Comb. Theory, Ser. {A}}, 82(1):112--120, 1998.

\bibitem{FranekFSS12}
Frantisek Franek, Robert C.~G. Fuller, Jamie Simpson, and William~F. Smyth.
\newblock More results on overlapping squares.
\newblock {\em J. Discrete Algorithms}, 17:2--8, 2012.

\bibitem{GawrychowskiKRW15}
Pawel Gawrychowski, Tomasz Kociumaka, Wojciech Rytter, and Tomasz Walen.
\newblock Tight bound for the number of distinct palindromes in a tree.
\newblock In {\em {SPIRE}}, volume 9309 of {\em Lecture Notes in Computer
  Science}, pages 270--276. Springer, 2015.

\bibitem{HegedusN14}
L{\'{a}}szl{\'{o}} Heged{\"{u}}s and Benedek Nagy.
\newblock Representations of circular words.
\newblock In {\em {AFL}}, volume 151 of {\em {EPTCS}}, pages 261--270, 2014.

\bibitem{Ilie05}
Lucian Ilie.
\newblock A simple proof that a word of length \emph{n} has at most 2\emph{n}
  distinct squares.
\newblock {\em J. Comb. Theory, Ser. {A}}, 112(1):163--164, 2005.

\bibitem{Ilie2007note}
Lucian Ilie.
\newblock A note on the number of squares in a word.
\newblock {\em Theoretical Computer Science}, 380(3):373--376, 2007.

\bibitem{KopylovaS12}
Evguenia Kopylova and William~F. Smyth.
\newblock The three squares lemma revisited.
\newblock {\em J. Discrete Algorithms}, 11:3--14, 2012.

\bibitem{Lothaire2002}
M.~{Lothaire}, editor.
\newblock {\em Algebraic Combinatorics on Words}, volume~90 of {\em
  Encyclopedia of Mathematics and its Applications}.
\newblock Cambridge University Press, Cambridge, 2002.

\bibitem{ManeaS15}
Florin Manea and Shinnosuke Seki.
\newblock Square-density increasing mappings.
\newblock In {\em {WORDS}}, volume 9304 of {\em Lecture Notes in Computer
  Science}, pages 160--169. Springer, 2015.

\bibitem{MasseBGL11}
Alexandre~Blondin Mass{\'{e}}, Srecko Brlek, Ariane Garon, and S{\'{e}}bastien
  Labb{\'{e}}.
\newblock Equations on palindromes and circular words.
\newblock {\em Theor. Comput. Sci.}, 412(27):2922--2930, 2011.

\bibitem{Shur10}
Arseny~M. Shur.
\newblock On ternary square-free circular words.
\newblock {\em Electr. J. Comb.}, 17(1), 2010.

\bibitem{Simpson07}
Jamie Simpson.
\newblock Intersecting periodic words.
\newblock {\em Theor. Comput. Sci.}, 374(1-3):58--65, 2007.

\bibitem{Simpson14}
Jamie Simpson.
\newblock Palindromes in circular words.
\newblock {\em Theor. Comput. Sci.}, 550:66--78, 2014.

\bibitem{thue1906}
A.~Thue.
\newblock {\"U}ber unendliche zeichenreihen.
\newblock {\em Norske Vid. Selsk. Skr., I Mat.--Nat. Kl., Christiania},
  7:1--22, 1906.

\end{thebibliography}

\end{document}